\documentclass[onecolumn, 11pt, letterpaper]{article}

\usepackage[linesnumbered,ruled,noend,nosemicolon]{algorithm2e}
\usepackage{amsmath}\allowdisplaybreaks
\usepackage{amssymb}
\usepackage{amsthm}
\usepackage[margin = 1in]{geometry}
\usepackage{hyperref}
\usepackage{mathtools}
\usepackage{authblk}

\usepackage{graphicx}
\usepackage{subfigure}

\newtheorem{corollary}{Corollary}
\newtheorem{lemma}{Lemma}
\newtheorem{theorem}{Theorem}

\newcommand{\E}{\mathbb{E}} 

\title{Efficient Deterministic Algorithms for Maximizing Symmetric Submodular Functions}
\author[1,2]{Zongqi Wan}
\author[1,2]{Jialin Zhang}
\author[1,2]{Xiaoming Sun}
\author[3]{Zhijie Zhang\thanks{Corresponding Author. Postal address: Fuzhou University, No.~2 Xueyuan Road, Minhou, Fuzhou, Fujian, China; E-mail: zzhang@fzu.edu.cn}}
\affil[1]{\small State Key Lab of Processors, Institute of Computing Technology, Chinese Academy of Sciences}
\affil[2]{\small School of Computer Science and Technology, University of Chinese Academy of Sciences}
\affil[3]{\small Center for Applied Mathematics of Fujian Province, School of Mathematics and Statistics, Fuzhou University}

\date{ }

\begin{document}

\maketitle

\begin{abstract}
    Symmetric submodular maximization is an important class of combinatorial optimization problems, including MAX-CUT on graphs and hyper-graphs. The state-of-the-art algorithm for the problem over general constraints has an approximation ratio of $0.432$~\cite{Feldman17}. The algorithm applies the canonical continuous greedy technique that involves a sampling process. It, therefore, suffers from high query complexity and is inherently randomized. In this paper, we present several efficient deterministic algorithms for maximizing a symmetric submodular function under various constraints. Specifically, for the cardinality constraint, we design a deterministic algorithm that attains a $0.432$ ratio and uses $O(kn)$ queries. Previously, the best deterministic algorithm attains a $0.385-\epsilon$ ratio and uses $O\left(kn (\frac{10}{9\epsilon})^{\frac{20}{9\epsilon}-1}\right)$ queries~\cite{chen2024guided}. For the matroid constraint, we design a deterministic algorithm that attains a $1/3-\epsilon$ ratio and uses $O(kn\log \epsilon^{-1})$ queries. Previously, the best deterministic algorithm can also attains $1/3-\epsilon$ ratio but it uses much larger $O(\epsilon^{-1}n^4)$ queries~\cite{LeeMNS09}. For the packing constraints with a large width, we design a deterministic algorithm that attains a $0.432-\epsilon$ ratio and uses $O(n^2)$ queries. To the best of our knowledge, there is no deterministic algorithm for the constraint previously. The last algorithm can be adapted to attain a $0.432$ ratio for single knapsack constraint using $O(n^4)$ queries. Previously, the best deterministic algorithm attains a $0.316-\epsilon$ ratio and uses $\widetilde{O}(n^3)$ queries~\cite{AmanatidisBM20}.

    \vspace{3mm}
\end{abstract}

\section{Introduction}

Submodular set functions have non-increasing marginal values as the set gets larger, which captures the effect of diminishing returns in reality.
The maximization of submodular functions is one of the central topics in combinatorial optimization.
It has found numerous applications, including viral marketing \cite{KempeKT03}, data summarization \cite{DasguptaKR13,TschiatschekIWB14}, causal inference \cite{SussexUK21,ZhouS16}, facility location \cite{AgeevS99}, maximum bisection \cite{AustrinBG16}.
The study of the submodular maximization problem has attracted a lot of attention in the literature.
Currently, the best algorithm achieves a $0.401$ approximation ratio under a general down-closed constraint \cite{buchbinder2023constrained}.
On the other hand, no algorithm can achieve an approximation ratio better than $0.478$ under the cardinality constraint \cite{Qi22} and the matroid constraint \cite{GharanV11}.

When the objective function enjoys additional properties, the problem often admits better approximation.
A well-studied case is the monotone submodular function.
With this objective, the problem admits $1-e^{-1}$ approximation under a general down-closed constraint \cite{CalinescuCPV11}.
This ratio is optimal even under the cardinality constraint \cite{NemhauserW78}.

In the absence of monotonicity, symmetry can also lead to better approximation.
A symmetric submodular function satisfies that any set and its complement have the same function value.
Examples include cut functions over undirected graphs and mutual information functions.
The study of symmetric submodular functions has received much attention in the literature \cite{Fujishige83,Queyranne98,GoemansS13,Feldman17,AmanatidisBM20,ChandrasekaranC21}.
Currently, the best algorithm attains a $0.432$ approximation ratio for maximizing symmetric submodular functions under a general down-closed constraint \cite{Feldman17}.

The algorithm in \cite{Feldman17} is based on the continuous greedy technique that involves a sampling process \cite{CalinescuCPV11,FeldmanNS11}.
Consequently, it is inherently randomized and suffers from high query complexity.
This hinders its applications in the real world and there is a demand for designing efficient deterministic algorithms for symmetric submodular maximization.
However, most such algorithms in the literature are designed for general non-monotone submodular functions. 
To the best of our knowledge, we only found two deterministic algorithms that are designed for symmetric functions. One is a $1/3-\epsilon$ approximation algorithm for matroid constraint~\cite{LeeMNS09}, which uses $O(\epsilon^{-1}n^4)$ queries. Another is the $0.316-\epsilon$ approximation algorithm for the knapsack constraint~\cite{AmanatidisBM20}, which uses $\widetilde{O}(n^3)$ queries. For the cardinality constraint, the best deterministic algorithm attains a $0.385-\epsilon$ ratio and uses $O\left(kn (\frac{10}{9\epsilon})^{\frac{20}{9\epsilon}-1}\right)$ queries~\cite{chen2024guided}. This algorithm is designed for general non-monotone submodular functions.


\subsection{Our Contributions}

In this paper, we present several efficient deterministic algorithms with improved approximation ratios or query complexity for maximizing symmetric submodular functions under various constraints.
Specifically,
\begin{enumerate}
	\item For the cardinality constraint, we design a deterministic algorithm that attains an approximation ratio of $\frac{1}{2}(1-e^{-2})\approx0.432$ and uses $O(kn)$ queries, where $n$ is the number of elements and $k$ is the size constraint.
	We also give a tight example showing that $\frac{1}{2}(1-e^{-2})$ is the best ratio that our algorithm can achieve.
	To accelerate the algorithm, we further design a randomized algorithm that attains an approximation ratio of $0.432-\epsilon$ and uses $O(k^2+n\log\epsilon^{-1})$ queries.
	Note that the algorithm is linear when $k=O(\sqrt{n})$.
	\item For the matroid constraint, we design a deterministic algorithm that attains an approximation ratio of $\frac{1}{3}-\epsilon$ and uses $O(kn\log \epsilon^{-1})$ queries, where $k$ is the rank of the matroid. Our algorithm achieves the same approximation ratio as the previous best deterministic algorithm~\cite{LeeMNS09}, but with a considerably lower query complexity.
	\item For the packing constraints with a large width, we present a deterministic algorithm that attains an approximation ratio of $0.432-\epsilon$ and uses $O(n^2)$ queries.
	For the knapsack constraint, the algorithm can be adapted to attain an approximation ratio of $0.432$ and use $O(n^4)$ queries.
\end{enumerate}
For the sake of comparison, we list the previous and our results in Table \ref{tab: main}.


\begin{table}[tb]
	\centering
	\begin{tabular}{ccccc}
		\hline
		Constraint & Function & Algo.~Ratio & Complexity & Type \\ \hline\hline
		General & Non-monotone & $0.401$ \cite{buchbinder2023constrained} & $\mathrm{poly}(n)$ & Rand \\ 
		General & Symmetric & $0.432$ \cite{Feldman17} & $\mathrm{poly}(n)$ & Rand \\ \hline\hline
		Cardinality & Non-monotone & $0.385-\epsilon$ \cite{chen2024guided} & $O\left(kn (\frac{10}{9\epsilon})^{\frac{20}{9\epsilon}-1}\right)$ & Det \\
		Cardinality & Symmetric & $0.432$ (Thm.~\ref{thm: cardinality}) & $O(kn)$ & Det \\
		Cardinality & Symmetric & $0.432-\epsilon$ (Thm.~\ref{thm: cardinality-random}) & $O(k^2+n\log\epsilon^{-1})$ & Rand \\ 
		\hline\hline
		Matroid & Non-monotone & $0.305-\epsilon$ \cite{chen2024guided} & $O\left(kn (\frac{10}{9\epsilon})^{\frac{20}{9\epsilon}-1}\right)$ & Det \\
            Matroid & Symmetric & $1/3-\epsilon$  \cite{LeeMNS09} & $O\left(\epsilon^{-1} n^4\right)$ & Det \\
		Matroid & Symmetric & $1/3-\epsilon$ (Thm.~\ref{thm: matroid}) & $O(kn\log \frac{n}{\epsilon})$ & Det \\ \hline\hline
		Packing & Symmetric & $0.432-\epsilon$ (Thm.~\ref{thm: packing}) & $O(n^2)$ & Det \\ \hline\hline
		Knapsack & Non-monotone & $0.25$ \cite{SunZZZ22} & $O(n^4)$ & Det \\
	    Knapsack & Symmetric & $0.316-\epsilon$ \cite{AmanatidisBM20} & $O(\frac{1}{\epsilon}n^3\log n)$ & Det \\
		Knapsack & Symmetric & $0.432$ (Thm.~\ref{thm: knapsack}) & $O(n^4)$ & Det \\
		\hline
	\end{tabular}
	\caption{Approximation algorithms for (symmetric) non-monotone submodular maximization under various constraints. ``Complexity'' refers to query complexity.
		``Rand'' is short for ``Randomized'' and ``Det'' is short for ``Deterministic''.}
	\label{tab: main}
\end{table}


\subsection{Related Work}

Early studies of submodular maximization, which date back to 1978, mainly focus on monotone submodular functions.
It was shown that no algorithm can achieve an approximation ratio better than $1-e^{-1}$ \cite{NemhauserW78}.
On the algorithmic side, a canonical greedy algorithm can achieve the optimal $1-e^{-1}$ ratio under the cardinality constraint \cite{NemhauserWF78}.
For the knapsack constraint, by combining the enumeration technique with the greedy algorithm, one can also achieve the $1-e^{-1}$ approximation ratio~\cite{KhullerMN99,Sviridenko04}.
Later, for the packing constraints with a large width, there is a multiplicative-updates algorithm that achieves the $1-e^{-1}$ approximation ratio \cite{AzarG12}.
For the matroid constraint, the greedy algorithm only attains a $1/2$ approximation ratio~\cite{FisherNW78}.
Vondr{\'{a}}k \cite{Vondrak08} made a breakthrough in 2008 by proposing the famous continuous greedy algorithm, which achieves a $1-e^{-1}$ approximation ratio under the matroid constraint.
The continuous greedy algorithm was later generalized to work for general down-closed constraints \cite{FeldmanNS11}.

For constrained non-monotone submodular maximization, no algorithm can achieve an approximation ratio better than $0.478$ under the cardinality constraint \cite{Qi22} and the matroid constraint \cite{GharanV11}.
The continuous greedy algorithm achieves $1/e$ approximation for this problem \cite{FeldmanNS11}.
After a series of works \cite{BuchbinderFNS14,EneN16,BuchbinderF19,tukan2024practical}, the best algorithm achieves a $0.401$ approximation ratio under a general down-closed constraint \cite{buchbinder2023constrained}.
The algorithm is randomized and suffers from high query complexity.
For the cardinality constraint, Buchbinder and Feldman~\cite{BuchbinderF18} designed a deterministic algorithm with $1/e$-approximation ratio and uses $O(k^3n)$ queries. Recently, Chen et al.~\cite{chen2024guided} provided a deterministic $(0.385-\epsilon)$-approximation algorithm with $O(kn (\frac{10}{9\epsilon})^{\frac{20}{9\epsilon}-1})$ queries. They also provided a deterministic $(0.305-\epsilon)$-approximation for the matroid constraint with the same query complexity.
For the knapsack constraint, the best deterministic algorithm attains a $1/4$ ratio and uses $O(n^4)$ queries \cite{SunZZZ22}.

Feige et al.~\cite{feige2011maximizing} first utilize the property of symmetric submodular function, they present a $(1/2-\epsilon)$-approximation algorithm for the unconstrained maximization problem of symmetric submodular function. Lee et al.~\cite{LeeMNS09} show a deterministic $(1/3-\epsilon)$-approximation for maximizing symmetric functions over matroid constraint. Both results are attained by the local search algorithm and they are not specially designed for the symmetric function. They only leverage the symmetric property in their analysis but do not exploit this property to modify the algorithm. Their algorithms also suffer from a high $O(\epsilon^{-1}n^4)$ query complexity. Feldman~\cite{Feldman17} was the first to design a specialized algorithm for symmetric submodular maximization over general downward close constraints. Using the continuous greedy technique, Feldman presented a $0.432$-approximation algorithm. In the case of the knapsack constraint, there exists a deterministic algorithm for symmetric submodular functions that achieves a $0.316-\epsilon$ ratio and utilizes $\widetilde{O}(n^3)$ queries \cite{AmanatidisBM20}.

\subsection{Paper Structure}

In Section \ref{sec: pre}, we introduce the symmetric submodular maximization problem formally.
In Section \ref{sec: cardinality}, we present our results for the problem under the cardinality constraint.
In Section \ref{sec: matroid}, we present our results for the problem under the matroid constraint.
In Section \ref{sec: packing}, we present our results for the problem under the packing constraints.
Finally, we conclude our paper in Section \ref{sec: conclusion}.

\section{Preliminaries}
\label{sec: pre}

Let $N$ be a finite ground set of $n$ elements.
For each set $S\subseteq N$ and element $u\in N$, we use $S+u$ to denote the union $S\cup\{u\}$ and $S-u$ to denote the difference set $S\setminus\{u\}$.
We also use $\bar{S}$ to denote the complementary set $N\setminus S$.
Let $f:2^N\rightarrow\mathbb{R}$ be a set function defined over $N$.
For any sets $S,T\subseteq N$ and element $u\in N$, we define $f(u\mid S)\coloneqq f(S+u)-f(S)$ and $f(T\mid S)\coloneqq f(S\cup T)-f(S)$.
$f$ is \emph{non-negative} if for any $S\subseteq N$, $f(S)\geq 0$.
$f$ is \emph{symmetric} if for any $S\subseteq N$, $f(S)=f(\bar{S})$.
$f$ is \emph{submodular} if for any $S,T\subseteq N$, $f(S)+f(T)\geq f(S\cup T)+f(S\cap T)$.
Equivalently, $f$ is submodular if for any $S\subseteq T$ and $u\notin T$, $f(S\cup\{u\})-f(S)\geq f(T\cup\{u\})-f(T)$.
Finally, we define $[n]\coloneqq \{1,2,\ldots,n\}$.

The symmetric submodular maximization problem can be formulated as $\max\{f(S):S\in\mathcal{I}\}$, where $f$ is the objective function and $\mathcal{I}$ is the constraint specifying the collection of feasible sets.
The function $f$ is assumed to be non-negative, symmetric and submodular.
Any algorithm for the problem can access the function via a \emph{value oracle}, which returns the function value $f(S)$ when set $S\subseteq N$ is queried.
The efficiency of the algorithm is measured by the query complexity, i.e.~the number of queries to the oracle.
The quantity should be polynomial in $|N|=n$.
Some well-studied constraints in the literature include:
\begin{itemize}
	\item \emph{Cardinality} constraint. For some $k\in\mathbb{N}_+$, $\mathcal{I}=\{S:|S|\leq k\}$.
	\item \emph{Matroid} constraint.
	A matroid system $\mathcal{M}=(N,\mathcal{I})$ consists of a finite ground set $N$ and a collection $\mathcal{I}\subseteq 2^N$ of the subsets of $N$, which satisfies a) $\emptyset \in \mathcal{I}$; b) if $A\subseteq B$ and $B\in \mathcal{I}$, then $A\in \mathcal{I}$; and c) if $A,B\in\mathcal{I}$ and $|A|<|B|$, then there exists an element $u\in B\setminus A$ such that $A \cup \{u\} \in \mathcal{I}$.
	A matroid constraint $\mathcal{I}$ is such that $\mathcal{M}=(N,\mathcal{I})$ forms a matroid system.
	For $\mathcal{M}$, each $A\in \mathcal{I}$ is called an \emph{independent set}.
	If $A$ is inclusion-wise maximal, it is called a \emph{base}.
	All bases of a matroid have an equal size, known as the \emph{rank} of the matroid.
	In this paper, we use $k$ to denote the rank.
	\item \emph{Packing} constraints.
	Given a matrix $A\in[0,1]^{m\times n}$ and a vector $b\in [1,\infty)^m$, the constraint can be written as $\mathcal{I}=\{S:Ax_S\leq b\}$, where $x_S$ represents the characteristic vector of $S$.
	When $m=1$, the constraint reduces to the canonical \emph{knapsack} constraint.
	The \emph{width} of the constraint is defined as $W=\min\{b_i/A_{ij}:A_{ij}>0\}$.
\end{itemize}

%
%


\subsection{Properties of Symmetric Submodular Functions}

The following lemma exploits the properties of the function and plays a central role in our algorithm design.
It is used for lower bounding $f(S\cup T)$ in terms of $f(T)$ in the absence of monotonicity.
A similar version for the multilinear extension of the function can be found in \cite{Feldman17}.

\begin{lemma}
	\label{lem: symmetric}
	Given a non-negative symmetric submodular function $f:2^N\rightarrow\mathbb{R}_+$ and a set $S\subseteq N$ such that $f(R)\leq f(S)$ for any $R\subseteq S$, then $f(S\cup T)\geq f(T)-f(S)$ for any $T\subseteq N$.
\end{lemma}

\begin{proof}
	Since $f$ is non-negative, symmetric and submodular,
	\[ f(T)-f(S\cup T)=f(\bar{T})-f(\bar{S}\cap\bar{T})\leq f(\bar{T}\setminus\bar{S})-f(\emptyset)\leq f(\bar{T}\cap S)\leq f(S). \]
	The last inequality follows from the condition of the lemma.
\end{proof}

At first glance, the condition of Lemma \ref{lem: symmetric} looks difficult to meet.
Nonetheless, we introduce the \textsc{Delete} procedure, described as Algorithm \ref{alg: delete}, to turn any set $S$ into one that satisfies the condition.
The procedure can be regarded as the ``delete'' operation of the local search algorithm, see e.g.~\cite{LeeMNS09}.

\begin{algorithm}[tb]
	\SetAlgoLined
	\KwIn{set $S\subseteq N$.}
	
	\ForEach($\backslash\backslash$ In arbitrary order){$u\in S$}{
		\lIf{$f(u\mid S-u)<0$}{$S\gets S-u$.}
	}
    \Return{$S$.}
    \caption{\textsc{Delete}}
    \label{alg: delete}
\end{algorithm}

\begin{lemma}
	\label{lem: removal}
	Given a non-negative symmetric submodular function $f:2^N\rightarrow\mathbb{R}_+$ and a set $S$ returned by Algorithm \ref{alg: delete}, then $f(R)\leq f(S)$ for any $R\subseteq S$.
\end{lemma}

\begin{proof}
	Let $S'$ be the value of $S$ before Algorithm \ref{alg: delete} and $S\setminus R=\{v_1,v_2,\ldots,v_{\ell}\}$, where $v_{j'}$ is visited before $v_j$ for $j'<j$ in Algorithm \ref{alg: delete}.
	Let $S_{v_j}$ be the value of $S'$ just before $v_j$ is visited.
	Clearly, $S\subseteq S_{v_j}\subseteq S'$ and hence $R\cup\{v_{j+1},v_{j+2},\ldots, v_{\ell}\}\subseteq S-v_j\subseteq S_{v_j}-v_j$.
	By submodularity,
	\[ f(S)=f(R)+\sum_{v_j\in S\setminus R} f(v_j\mid R\cup\{v_{j+1},v_{j+2},\ldots, v_{\ell}\})\geq f(R)+\sum_{v_j\in S\setminus R} f(v_j\mid S_{v_j}-v_j). \]
	Finally, observe that $v_j\in S$ implies that it is not removed in Algorithm \ref{alg: delete}.
	It follows that $f(v_j\mid S_{v_j}-v_j)\geq 0$ and therefore $f(S)\geq f(R)$.
\end{proof}

By combining Lemmas \ref{lem: symmetric} and \ref{lem: removal}, we have

\begin{corollary}
	\label{cor: symmetric}
	Given a non-negative symmetric submodular function $f:2^N\rightarrow\mathbb{R}_+$ and a set $S$ returned by Algorithm \ref{alg: delete}, then $f(S\cup T)\geq f(T)-f(S)$ for any $T\subseteq N$.
\end{corollary}

\section{Cardinality Constraint}
\label{sec: cardinality}

In this section, we present two algorithms for maximizing symmetric submodular functions under the cardinality constraint $\mathcal{I}=\{S:|S|\leq k\}$.
In Section \ref{sec: greedy cardinality}, we present a deterministic algorithm that has an approximation ratio of $\frac{1}{2}(1-e^{-2})\approx 0.432$ and uses $O(nk)$ queries.
In Section \ref{sec: sample greedy cardinality}, we present a fast randomized algorithm that has an approximation ratio of $\frac{1}{2}(1-e^{-2(1-\epsilon)})\approx 0.432-\epsilon$ and uses $O(k^2+n\log\epsilon^{-1})$ queries.
In Section \ref{sec: tight example}, we present a tight example for our deterministic algorithm.

\subsection{The Greedy Algorithm}
\label{sec: greedy cardinality}

In this section, we present a greedy algorithm for maximizing a symmetric submodular function under the cardinality constraint.
The overall procedure is depicted as Algorithm \ref{alg: cardinality}.
Compared with the canonical greedy algorithm for maximizing \emph{monotone} submodular functions, Algorithm \ref{alg: cardinality} has an extra step that executes the \textsc{Delete} procedure to update $S_i$ immediately after the addition of $u_i$ in each round.
This helps lower bound $f(S_i\cup O)$ in the absence of monotonicity, where $O\in\arg\max\{f(S):|S|\leq k\}$.
We now analyze Algorithm \ref{alg: cardinality} by a standard argument.

\begin{algorithm}[tb]
	\SetAlgoLined
	\KwIn{ground set $N$, objective function $f$, cardinality $k$.}
	
	$S_0 \gets \emptyset$.
	
	\For{$i=1$ to $k$}{
		$u_i\gets\arg\max_{u\in N} f(u\mid S_{i-1})$.
		
		$S_i\gets S_{i-1}+u_i$.
		
		$S_i\gets$ \textsc{Delete}$(S_i)$.
		
	}
	
	\Return $S_k$.
	\caption{\textsc{Greedy-Cardinality}}
	\label{alg: cardinality}
\end{algorithm}

\begin{lemma}
	\label{lem: marginal gain}
	For every $i\in [k]$, $f(S_i)-f(S_{i-1})\geq \frac{1}{k}\cdot(f(O)-2\cdot f(S_{i-1}))$.
\end{lemma}

\begin{proof}
	Let $S'_i=S_{i-1}+u_i$ be the value of $S_i$ just before \textsc{Delete} was executed. 
	By the choice of $u_i$, $f(u_i\mid S_{i-1})\geq f(o\mid S_{i-1})$ for any $o\in O$.
	Then,
	\begin{align*}
		f(S'_i)-f(S_{i-1}) &=f(u_i\mid S_{i-1})\geq \frac{1}{k}\sum_{o\in O} f(o\mid S_{i-1})\geq \frac{1}{k}\cdot f(O\mid S_{i-1}) \\
		&=\frac{1}{k}\cdot (f(S_{i-1}\cup O)-f(S_{i-1}))\geq \frac{1}{k}\cdot (f(O)-2\cdot f(S_{i-1})).
	\end{align*}
    The second inequality follows from submodularity and the last inequality follows from Corollary \ref{cor: symmetric}, since $S_{i-1}$ is the output of \textsc{Delete} in the last round.
    The lemma follows by observing that the value of $S_i$ only increases during the execution of \textsc{Delete} and hence $f(S_i)\geq f(S'_i)$.
\end{proof}

\begin{theorem}
	\label{thm: cardinality}
	Algorithm \ref{alg: cardinality} has an approximation ratio of $\frac{1}{2}(1-e^{-2})$ and uses $O(kn)$ queries.
\end{theorem}

\begin{proof}
	For the query complexity, note that Algorithm \ref{alg: cardinality} has $k$ rounds.
	In each round, the selection of $u_i$ needs $O(n)$ queries and the \textsc{Delete} procedure needs $O(k)$ queries since $S_i$ contains at most $k$ elements.
	Therefore, Algorithm \ref{alg: cardinality} uses $O(kn)$ queries in total.
	
	Next, we show by induction that for $i\in[k]$,
	\[ f(S_i)\geq \frac{1}{2}\left(1-\left(1-\frac{2}{k}\right)^i\right)\cdot f(O). \]
	The claim holds for $i=0$ since $f(\emptyset)\geq 0$.
	Assume that it holds for $i'<i$.
	By Lemma \ref{lem: marginal gain},
	\begin{align*}
		f(S_i) &\geq \left(1-\frac{2}{k}\right)\cdot f(S_{i-1})+\frac{1}{k}\cdot f(O) \\
		&\geq \left(1-\frac{2}{k}\right)\cdot\frac{1}{2}\left(1-\left(1-\frac{2}{k}\right)^{i-1}\right)\cdot f(O)+\frac{1}{k}\cdot f(O) \\
		&=\frac{1}{2}\left(1-\left(1-\frac{2}{k}\right)^i\right)\cdot f(O).
	\end{align*}
	By plugging $i=k$ and $e^x\geq 1+x$ for $x\in\mathbb{R}$, we have
	\[ f(S_k)\geq \frac{1}{2}\left(1-\left(1-\frac{2}{k}\right)^k\right)\cdot f(O)\geq \frac{1}{2}(1-e^{-2})\cdot f(O). \]
\end{proof}

\subsection{The Sample Greedy Algorithm}
\label{sec: sample greedy cardinality}

In this section, we present a randomized algorithm for maximizing symmetric submodular functions under the cardinality constraint.
By using randomness, its query complexity is reduced to $O(k^2+n\log\epsilon^{-1})$, which is linear when $k=O(\sqrt{n})$.
Our algorithm is based on the Sample Greedy algorithm \cite{MirzasoleimanBK15,BuchbinderFS15}, with an extra step that executes the \textsc{Delete} procedure to update $S_i$ immediately after the addition of $u_i$ in each round.
The formal procedure is depicted as Algorithm \ref{alg: sample greedy cardinality}.

Let $O\in\arg\max\{f(S):|S|\leq k\}$.
We now give an analysis of Algorithm \ref{alg: cardinality}.

\begin{algorithm}[tb]
	\SetAlgoLined
	\KwIn{ground set $N$, objective function $f$, cardinality $k$.}
	
	$S_0 \gets \emptyset$.
	
	\For{$i=1$ to $k$}{
		$R_i\gets$ a random subset obtained by sampling $r=\lceil\frac{n}{k}\ln\epsilon^{-1}\rceil$ random elements from $N\setminus S_{i-1}$.
		
		$u_i\gets\arg\max_{u\in R_i} f(u\mid S_{i-1})$.
		
		$S_i\gets S_{i-1}+u_{i}$.
		
		$S_i\gets$ \textsc{Delete}$(S_i)$.	
	}
	
	\Return $S_k$.
	\caption{\textsc{Sample-Greedy-Cardinality}}
	\label{alg: sample greedy cardinality}
\end{algorithm}

\begin{lemma}\label{lem:sample greedy lemma}
	For every $i\in [k]$, $\E[f(S_i)]-\E[f(S_{i-1})]\geq\frac{1-\epsilon}{k}\cdot(f(O)-2\cdot \E[f(S_{i-1})])$.
\end{lemma}

\begin{proof}
	Fix $i\in[k]$ and all the random choices of Algorithm \ref{alg: sample greedy cardinality} up to round $i-1$.
	Then, $S_{i-1}$ is determined.
	Conditioned on this, consider the random choice in round $i$.
	Let us estimate the probability that $R_i\cap (O\setminus S_{i-1})\neq \emptyset$.
	By the construction of $R_i$,
	\[ \Pr[R_i\cap (O\setminus S_{i-1})=\emptyset]=\left(1-\frac{|O\setminus S_{i-1}|}{|N\setminus S_{i-1}|}\right)^{r}\leq\exp\left(-\frac{r}{n}|O\setminus S_{i-1}|\right). \]
	The inequality follows from $1-x\leq e^{-x}$ for $x\in\mathbb{R}$ and $|N\setminus S_{i-1}|\leq n$.
	Next, by the concavity of $1-e^{-\frac{r}{n}x}$ as a function of $x$ and the fact that $x=|O\setminus S_{i-1}|\in[0,k]$, we have
	\[ \Pr[R_i\cap (O\setminus S_{i-1})\neq \emptyset]\geq 1-\exp\left(-\frac{r}{n}|O\setminus S_{i-1}|\right)\geq\left(1-e^{-\frac{rk}{n}}\right)\frac{|O\setminus S_{i-1}|}{k}\geq (1-\epsilon)\frac{|O\setminus S_{i-1}|}{k}. \]
	The last inequality is due to the value of $r$.
	
	Let $S'_i=S_{i-1}+u_i$ be the value of $S_i$ just before \textsc{Delete} was executed. 
	Under the event that $R_i\cap (O\setminus S_{i-1})\neq \emptyset$, the marginal value of $u_i$ is at least that of a uniformly random element from $R_i\cap (O\setminus S_{i-1})$.
	Since $R_i$ contains each element of $O\setminus S_{i-1}$ with equal probability, a uniformly random element from $R_i\cap (O\setminus S_{i-1})$ is actually a uniformly random element from $(O\setminus S_{i-1})$.
	Thus, we have
	\begin{align*}
		\E[f(S'_i)-f(S_{i-1})]&=\E[f(u_i\mid S_{i-1})] \geq \Pr[R_i\cap (O\setminus S_{i-1})\neq \emptyset]\times\frac{1}{|O\setminus S_{i-1}|}\sum_{o\in O\setminus S_{i-1}}f(o\mid S_{i-1}) \\
		&\geq \frac{1-\epsilon}{k} (f(S_{i-1}\cup O)-f(S_{i-1}))\geq \frac{1-\epsilon}{k} (f(O)-2\cdot f(S_{i-1})).
	\end{align*}
	The second inequality is due to submodularity and the last follows from Corollary \ref{cor: symmetric}, since $S_{i-1}$ is the output of \textsc{Delete} in the last round.
	By taking the randomness of $S_{i-1}$, we have
	\[ \E[f(S'_i)-f(S_{i-1})]\geq \frac{1-\epsilon}{k} (f(O)-2\cdot \E[f(S_{i-1})]). \]
	Finally, observe that the value of $S_i$ only increases during the execution of \textsc{Delete} and hence $f(S_i)\geq f(S'_i)$.
	By taking the randomness of $S_{i-1}$ and $u_i$, we have $\E[f(S_i)]\geq \E[f(S'_i)]$.
	The lemma follows immediately.
\end{proof}

\begin{theorem}
	\label{thm: cardinality-random}
	Algorithm \ref{alg: sample greedy cardinality} has an approximation ratio of $\frac{1}{2}(1-e^{-2(1-\epsilon)})$ and uses $O(k^2+n\log\epsilon^{-1})$ queries.
\end{theorem}

\begin{proof}
	For the query complexity, note that Algorithm \ref{alg: sample greedy cardinality} has $k$ rounds.
	In each round, it costs $O(\frac{n}{k}\log\epsilon^{-1})$ queries to find $u_i$ and the \textsc{Delete} Procedure needs $O(k)$ queries since $S$ contains at most $k$ elements.
	Therefore, Algorithm \ref{alg: cardinality} uses $O(k^2+n\log\epsilon^{-1})$ queries in total.
	
	Next, we show by induction that for every $i\in[k]$,
	\[ \E[f(S_i)]\geq \frac{1}{2}\left(1-\left(1-\frac{2(1-\epsilon)}{k}\right)^i\right)\cdot f(O). \]
	The claim holds for $i=0$ since $f(\emptyset)\geq 0$.
	Assume that it holds for $i'<i$.
	By Lemma \ref{lem: marginal gain},
	\begin{align*}
		\E[f(S_i)] &\geq \left(1-\frac{2(1-\epsilon)}{k}\right)\cdot \E[f(S_{i-1})]+\frac{1}{k}\cdot f(O) \\
		&\geq \left(1-\frac{2(1-\epsilon)}{k}\right)\cdot\frac{1}{2}\left(1-\left(1-\frac{2(1-\epsilon)}{k}\right)^{i-1}\right)\cdot f(O)+\frac{1}{k}\cdot f(O) \\
		&\geq\frac{1}{2}\left(1-\left(1-\frac{2(1-\epsilon)}{k}\right)^i\right)\cdot f(O).
	\end{align*}
	By plugging $i=k$ and $e^x\geq 1+x$ for $x\in\mathbb{R}$, we have
	\[ \E[f(S_k)]\geq \frac{1}{2}\left(1-\left(1-\frac{2(1-\epsilon)}{k}\right)^k\right)\cdot f(O)\geq \frac{1-e^{-2(1-\epsilon)}}{2}\cdot f(O). \]
\end{proof}

\subsection{A Tight Example for Algorithm \ref{alg: cardinality}}
\label{sec: tight example}

In this section, we give a tight example for Algorithm \ref{alg: cardinality}, showing that our analysis for it is tight.

Consider an instance of MAX-CUT on a bipartite graph $G=(O,S\cup T, E)$ as follows.
Assume that $k\geq 3$.
Let $O=\{o_1,o_2,\ldots,o_k\}$, $S=\{u_1,u_2,\ldots,u_k\}$, and $T=\{v_{ij}\}_{1\leq i\leq k,1\leq j\leq c}$ for some constant $c$ satisfying
\[ c=\left\lceil\frac{1+(1-2/k)^k}{2(1-2/k)^{k-1}}\right\rceil. \]
For $i,j\in[k]$, $(o_i,u_j)\in E$ and its weight is $\frac{1}{k}\left(1-\frac{2}{k}\right)^{j-1}$.
For $i\in [k]$ and $j\in[c]$, $(o_i, v_{ij})\in E$ and its weight is $\frac{1+(1-2/k)^k}{2c}$.
Then, $f(o_i)=\sum_{j=1}^{k}\frac{1}{k}\left(1-\frac{2}{k}\right)^{j-1}+c\cdot \frac{1+(1-2/k)^k}{2c}=1$.
The optimal solution is $O$ and $f(O)=\sum_{i=1}^{k}f(o_i)=k$.
We will argue that Algorithm \ref{alg: cardinality} may pick $u_j$ in round $j$ and finally return $S$.
Since $f(S)=\sum_{j=1}^{k}\left(1-\frac{2}{k}\right)^{j-1}=\frac{k}{2}\left(1-(1-2/k)^k\right)=\frac{1}{2}\left(1-(1-2/k)^k\right)\cdot f(O)$, this proves our claim.

Define $S_0=\emptyset$ and $S_j=\{u_1,u_2,\ldots,u_j\}$ for $j\in [k]$.
Assume that Algorithm \ref{alg: cardinality} selects $S_{j-1}=\{u_1,u_2,\ldots,u_{j-1}\}$ before round $j$.
In round $j$, for $i\geq j$, $f(u_i\mid S_{j-1})=(1-2/k)^{i-1}$.
For $i\in [k]$,
$f(o_i\mid S_{j-1})=1-2\sum_{j'=1}^{j-1}\frac{1}{k}(1-2/k)^{j'-1}=(1-2/k)^{j-1}$.
For $i\in [k],j'\in[c]$,
By the definition of $c$, $f(v_{ij'}\mid S_{j-1})=f(v_{ij'})=\frac{1+(1-2/k)^k}{2c}\leq (1-2/k)^{k-1}\leq (1-2/k)^{j-1}$.
Therefore, only $u_j$ and elements in $O$ attain the maximum marginal gain $(1-2/k)^{j-1}$ in this round.
Thus, Algorithm \ref{alg: cardinality} may select $u_j$ during its execution.
Finally, it is easy to see that $f(u\mid S_j-u)\geq 0$ for any $u\in S_j$.
Hence, $S_j$ remains unchanged after the \textsc{Delete} procedure.

%
%
%
%
%
%
%
%
%
%
%

\section{Matroid Constraint}
\label{sec: matroid}

In this section, we present a greedy-based algorithm for maximizing a symmetric submodular function under the matroid constraint $\mathcal{I}=\{S:S\in\mathcal{I}\}$, where $\mathcal{M}=(N,\mathcal{I})$ forms a matroid system.
The overall procedure is depicted as Algorithm \ref{alg: matroid}.

For ease of description, we turn the original instance $(N,f,\mathcal{M})$ to a new one $(N',f',\mathcal{M}')$ by adding a set $D$ of $2k$ ``dummy elements'' in the following way.
\begin{itemize}
	\item $N'=N\cup D$.
	\item $f'(S)=f(S\setminus D)$ for every set $S\subseteq N'$.
	\item $S\in\mathcal{I}'$ if and only if $S\setminus D\in\mathcal{I}$ and $|S|\leq k$.
\end{itemize}
Clearly, the new instance and the old one refer to the same problem.
By overloading notations, we still denote by $(N,f,\mathcal{M})$ the instance with dummy elements.
Another ingredient of Algorithm \ref{alg: matroid} is the well-known exchange property of matroids, which is stated as Lemma \ref{lem: exchange property}.
\begin{lemma}[\cite{Schrijver03}]
	\label{lem: exchange property}
	If $A$ and $B$ are two bases of a matroid $\mathcal{M}=(N,\mathcal{I})$, there exists a one-to-one function $g:A\rightarrow B$ such that
	\begin{itemize}
		\item $g(u)=u$ for every $u\in A\cap B$.
		\item for every $u\in A$, $B+u-g(u)\in\mathcal{I}$.
	\end{itemize}
\end{lemma}

Algorithm \ref{alg: matroid} starts with a set of $k$ dummy elements denoted by $S_0$.
It runs $K:=\lceil \frac{k}{3}\ln \epsilon^{-1}\rceil$ rounds in total to update $S_0$ and the solution by the end of round $i$ is denoted by $S_i$.
At round $i$, it finds a base $M_i\subseteq N\setminus S_{i-1}$ whose marginal value is maximized, a mapping $g_i$ between $M_i$ and $S_{i-1}$ defined in Lemma \ref{lem: exchange property}, and the element $u_i\in M_i$ that maximizes $f(S_{i-1}+u_i-g_i(u_i))-f(S_{i-1})$.
It then sets $S_i=S_{i-1}+u_i-g_i(u_i)$ and executes the \textsc{Delete} procedure to update $S_i$.
Finally, it returns $S_K$.

Note that the task of finding $M_i$ involves solving an additive maximization problem over a matroid. It is well known that this problem can be efficiently addressed using a simple greedy algorithm, which bears resemblance to Kruskal's algorithm. The algorithm begins by sorting the elements $\{f(u\mid S_{i-1})\}_{u\in N\backslash S_{i-1}}$ in descending order. Then, proceeding in this descending order, each element $u$ is added to the solution set if its inclusion does not violate the matroid constraint. Note that we can always extend the subset obtained by the above procedure to a matroid base by adding some dummy elements to it. Also, $g_i$ can be found by invoking an algorithm that finds a perfect matching in a bipartite graph.

\begin{algorithm}[t]
	\SetAlgoLined
	\KwIn{ground set $N$, objective function $f$, matroid $\mathcal{M}=(N,\mathcal{I})$.}
	
	$S_0\gets$ an arbitrary base containing only elements of $D$.

        $K\gets \lceil \frac{k}{3}\ln \epsilon^{-1}\rceil$
 
	\For{$i=1$ to $K$}{
		Let $M_i\subseteq N\setminus S_{i-1}$ be a base of $\mathcal{M}$ maximizing $\sum_{u\in M_i} f(u \mid S_{i-1})$.
		
		Let $g_i$ be the mapping defined in Lemma \ref{lem: exchange property} by plugging $A=M_i$ and $B=S_{i-1}$.
		
		$u_i\gets\arg\max_{u\in M_i} f(S_{i-1}+u-g_i(u))-f(S_{i-1})$.
		
		$S_i\gets S_{i-1}+u_i-g_i(u_i)$.
		
		$S_i\gets$ \textsc{Delete}$(S_i)$.
	}
	
	\Return $S_K$.
	\caption{\textsc{Greedy-Matroid}}
	\label{alg: matroid}
\end{algorithm}

We now give an analysis of Algorithm \ref{alg: matroid}.
Let $O=\arg\max\{f(S):S\in\mathcal{I}\}$, we have

\begin{lemma}
	\label{lem: marginal gain-matroid}
	For every $i\in [K]$, $f(S_i)-f(S_{i-1})\geq \frac{1}{k}\cdot(f(O)-3\cdot f(S_{i-1}))$.
\end{lemma}

\begin{proof}
	Observe that $O\setminus S_{i-1}$ plus enough dummy elements in $D\setminus S_{i-1}$ forms a valid candidate for $M_i$.
	By the construction of $M_i$,
	\begin{align}
		\frac{1}{k}\sum_{u\in M_i} f(u\mid S_{i-1}) &\geq\frac{1}{k} \sum_{u\in O\setminus S_{i-1}} f(u\mid S_{i-1})
		\geq \frac{1}{k}\cdot (f(O\cup S_{i-1})-f(S_{i-1})) \nonumber \\
		&\geq \frac{1}{k}\cdot (f(O)-2\cdot f(S_{i-1})). \label{eq: matroid-1}
	\end{align}
	The second inequality is due to submodularity.
	The last inequality follows from Corollary \ref{cor: symmetric}, since $S_{i-1}$ is the output of \textsc{Delete} in the last round.
	On the other hand, by submodularity and non-negativity,
	\begin{equation}
		\frac{1}{k}\sum_{u\in M_i}f(g_i(u)\mid S_{i-1}-g_i(u))\leq \frac{f(S_{i-1})-f(\emptyset)}{k}\leq \frac{f(S_{i-1})}{k}. \label{eq: matroid-2}
	\end{equation}
	Let $S'_i=S_{i-1}+u_i-g_i(u_i)$ be the value of $S_i$ just before \textsc{Delete} is executed.
	By the choice of $u_i$, we have
	\begin{align*}
		f(S'_i) &=f(S_{i-1}+u_i-g_i(u_i)) \\
		&\geq \frac{1}{k}\sum_{u\in M_i} f(S_{i-1}+u-g_i(u)) \\
		&\geq\frac{1}{k}\sum_{u\in M_i} (f(S_{i-1}+u)+f(S_{i-1}-g_i(u))-f(S_{i-1})) \\
		&=\frac{1}{k}\sum_{u\in M_i} (f(S_{i-1})+f(u\mid S_{i-1})-f(g_i(u)\mid S_{i-1}-g_i(u))) \\
		&\geq f(S_{i-1})+\frac{1}{k}\cdot(f(O)-2\cdot f(S_{i-1}))-\frac{f(S_{i-1})}{k} \\
		&=f(S_{i-1})+\frac{1}{k}\cdot(f(O)-3\cdot f(S_{i-1})).
	\end{align*}
	The second inequality is due to submodularity.
	The last inequality is by Eq.~\eqref{eq: matroid-1} and \eqref{eq: matroid-2}.
	Finally, observe that the value of $S_i$ only increases during the execution of \textsc{Delete} and hence $f(S_i)\geq f(S'_i)$.
	The lemma follows immediately.
\end{proof}

\begin{theorem}
	\label{thm: matroid}
	Algorithm \ref{alg: matroid} has an approximation ratio of $(1-\epsilon)/3$ and uses $O(kn\log \frac{n}{\epsilon})$ queries.
\end{theorem}

\begin{proof}
	For the query complexity, note that Algorithm \ref{alg: matroid} has $K=\lceil \frac{k}{3} \ln \epsilon^{-1}\rceil$ rounds.
	In each round, the construction of $M_i$ needs $O(n\log n)$ queries and the \textsc{Delete} procedure needs $O(k)$ queries since $S_i$ contains at most $k$ elements.
	Therefore, Algorithm \ref{alg: matroid} uses $O( k n\log \frac{n}{\epsilon})$ queries in total.
	
	Next, we show by induction that for $i\in[k]$,
	\[ f(S_i)\geq \frac{1}{3}\left(1-\left(1-\frac{3}{k}\right)^i\right)\cdot f(O). \]
	The claim holds for $i=0$ since $f(\emptyset)\geq 0$.
	Assume that it holds for $i'<i$.
	By Lemma \ref{lem: marginal gain-matroid},
	\begin{align*}
		f(S_i) &\geq \left(1-\frac{3}{k}\right)\cdot f(S_{i-1})+\frac{1}{k}\cdot f(O) \\
		&\geq \left(1-\frac{3}{k}\right)\cdot\frac{1}{3}\left(1-\left(1-\frac{3}{k}\right)^{i-1}\right)\cdot f(O)+\frac{1}{k}\cdot f(O) \\
		&=\frac{1}{3}\left(1-\left(1-\frac{3}{k}\right)^i\right)\cdot f(O).
	\end{align*}
	By plugging $i=K$ and $e^x\geq 1+x$ for $x\in\mathbb{R}$, we have
	\[ f(S_K)\geq \frac{1}{3}\left(1-\left(1-\frac{3}{k}\right)^{(k\ln \epsilon^{-1})/3}\right)\cdot f(O)\geq \frac{1-\epsilon}{3}\cdot f(O). \]
\end{proof}

\section{Packing Constraints}
\label{sec: packing}

In this section, we present a deterministic algorithm for maximizing symmetric submodular functions under the packing constraints with a large width.
Our algorithm is based on the multiplicative updates approach in \cite{AzarG12} for maximizing monotone submodular functions, with an extra step that executes the \textsc{Delete} procedure to update the current solution $S$ at the end of each iteration.
The complete procedure is depicted as Algorithm \ref{alg: packing}.
For ease of description, we identify $N$ with $[n]$ and use $j\in[n]$ to denote an element.

\begin{algorithm}[tb]
	\SetAlgoLined
	\KwIn{ground set $N$, objective function $f$, $\lambda>1$, matrix $A\in[0,1]^{m\times n}$ and vector $b\in [1,\infty)^m$.}
	
	$S\gets\emptyset$.
	
	\lFor{$i=1$ to $m$}{$w_i\gets 1/b_i$.}
	
	\While{$\sum_{i=1}^{m}b_i w_i\leq \lambda$ and $S\neq [n]$}{
		$j\gets\arg\max_{j\in N\setminus S} \frac{f(j\mid S)}{\sum_{i=1}^{m} A_{ij}w_i}$.
		
		\lIf{$f(j\mid S)\leq 0$}{\textbf{break}.}
		
		$S\gets S+j$.
		
		$S\gets$ \textsc{Delete}$(S)$.
		
		\lFor{$i=1$ to $m$}{$w_i\gets w_i\lambda^{A_{ij}/b_i}$.}
	}
	
	\Return $S$ if $S$ is feasible and $S-j^*$ otherwise, where $j^*$ is the last element added into $S$.
	\caption{\textsc{Multiplicative-Updates-Packing}}
	\label{alg: packing}
\end{algorithm}

We now give an analysis of Algorithm \ref{alg: packing}.
Assume that the algorithm added $t+1$ elements into $S$ in total.
For each $r\in[t+1]$, let $S_r$ be the value of $S$ at the end of round $r$, $j_r$ be the element selected in round $r$.
For each $i\in[m]$ and $r\in[t+1]$, let $w_{ir}$ be the value of $w_i$ at the end of round $r$ and $\beta_r=\sum_{i=1}^m b_i w_{ir}$.
Let $O\in\arg\max\{f(S):Ax_S\leq b\}$.

We first show that the algorithm outputs a feasible solution.

\begin{lemma}
	\label{lem: feasible}
	Algorithm \ref{alg: packing} outputs a feasible solution.
\end{lemma}
\begin{proof}
	By our notations, the algorithm returns $S_{t+1}$ if it is feasible and $S_t$ otherwise.
	The lemma clearly holds in the former case.
	For the latter case, let $\ell$ be the first element whose addition into $S$ leads to a violation in some constraint.
	That is, suppose $\ell$ was added in round $t'$, then $S_{t'-1}$ is feasible but for some $i\in[m]$, $\sum_{j\in S_{t'}} A_{ij}>b_i$.
	Then, we have
	\begin{align*}
		b_iw_{it'}=b_iw_{i0}\prod_{j\in S_{t'}}\lambda^{A_{ij}/b_i}=\lambda^{\sum_{j\in S_{t'}}A_{ij}/b_i}>\lambda.
	\end{align*}
    Thus, the \textbf{while} loop breaks immediately at the beginning of the next round.
    This means that $\ell$ is the last element added into $S$ and therefore $t'=t+1$.
    It follows that the returned solution $S_t$ is feasible.
\end{proof}
Next, we present two useful lemmas for our analysis.
\begin{lemma}
	\label{lem: expand recursion}
	Given a set function $f:2^N\rightarrow\mathbb{R}_+$, a collection of sets $S_0, S_1,\ldots, S_t\subseteq N$ satisfying $f(S_0)\leq f(S_1)\leq \ldots \leq f(S_t)$ and a set $O\subseteq N$ satisfying $f(O)>2 f(S_t)$, then
	\[ \sum_{r=1}^{t}\frac{f(S_{r})-f(S_{r-1})}{f(O)-2f(S_{r-1})}\leq \frac{1}{2}\ln \frac{f(O)-2f(S_0)}{f(O)-2f(S_t)}. \]
\end{lemma}

\begin{proof}
	For each $r\in [t]$, observe that
	\[ \frac{f(S_{r})-f(S_{r-1})}{f(O)-2f(S_{r-1})}=\int_{f(S_{r-1})}^{f(S_{r})}\frac{1}{f(O)-2f(S_{r-1})}\, dx\leq \int_{f(S_{r-1})}^{f(S_{r})}\frac{1}{f(O)-2x}\, dx. \]
	The inequality holds since $1/(f(O)-2x)$ is monotonically increasing for $x\in [0,f(O)/2)$.
	Consequently, we have
	\[ \sum_{r=1}^{t}\frac{f(S_{r})-f(S_{r-1})}{f(O)-2f(S_{r-1})}\leq  \sum_{r=1}^{t}\int_{f(S_{r-1})}^{f(S_{r})}\frac{1}{f(O)-2x}\, dx=\int_{f(S_0)}^{f(S_t)}\frac{1}{f(O)-2x}\, dx=\frac{1}{2}\ln \frac{f(O)-2f(S_0)}{f(O)-2f(S_t)}. \]
\end{proof}

\begin{lemma}
	\label{lem: marginal gain-packing}
	For every $r\in[t]$,
	\[ \frac{f(S_r)-f(S_{r-1})}{\sum_{i=1}^m A_{ij_r} w_{i(r-1)}}\geq\frac{f(O)-2f(S_{r-1})}{\beta_{r-1}}. \]
\end{lemma}

\begin{proof}
	Let $S'_r=S_{r-1}+j_r$ be the value of $S_r$ just before \textsc{Delete} was executed.
	By the choice of $j_r$, for any $j\in O\setminus S_{r-1}$,
	\[ \frac{f(j_r\mid S_{r-1})}{\sum_{i=1}^{m}A_{ij_r}w_{i(r-1)}}\geq\frac{f(j\mid S_{r-1})}{\sum_{i=1}^{m}A_{ij}w_{i(r-1)}}.  \]
	By summing over $j\in O\setminus S_{r-1}$, we have
	\begin{align*}
		\frac{f(S'_r)-f(S_{r-1})}{\sum_{i=1}^{m}A_{ij_r}w_{i(r-1)}}\sum_{j\in O\setminus S_{r-1}}\sum_{i=1}^{m}A_{ij}w_{i(r-1)} &\geq \sum_{j\in O\setminus S_{r-1}} f(j\mid S_{r-1})\geq f(O\mid S_{r-1}) \\
		&=f(S_{r-1}\cup O)-f(S_{r-1})\geq f(O)-2 f(S_{r-1}).
	\end{align*}
	The second inequality is due to submodularity and the last inequality follows from Corollary \ref{cor: symmetric}, since $S_{r-1}$ is the output of \textsc{Delete} in the last round.
	On the other hand,
	\[ \sum_{j\in O\setminus S_{r-1}}\sum_{i=1}^{m}A_{ij}w_{i(r-1)}=\sum_{i=1}^{m}\sum_{j\in O\setminus S_{r-1}} A_{ij}w_{i(r-1)}\leq \sum_{i=1}^{m} b_iw_{i(r-1)}=\beta_{r-1}. \]
	The inequality holds since $O\setminus S_{r-1}$ is feasible.
	Combining the two inequalities,
	\[ \frac{f(S'_r)-f(S_{r-1})}{\sum_{i=1}^{m}A_{ij_r}w_{i(r-1)}}\geq \frac{f(O)-2 f(S_{r-1})}{\beta_{r-1}}.  \]
	The lemma follows by observing that the value of $S_r$ only increases during the execution of \textsc{Delete} and hence $f(S_r)\geq f(S'_r)$.
\end{proof}

\begin{theorem}
	\label{thm: packing}
	Given $\epsilon\in(0,1)$ and assume that $W\geq\max\{\ln m,1\}/\epsilon^2$, by setting $\lambda=e^{\epsilon W}$,
	Algorithm \ref{alg: packing} has an approximation ratio of $\frac{1}{2}(1-e^{-2(1-3\epsilon)})$ and uses $O(n^2)$ queries.
\end{theorem}

\begin{proof}
	For the query complexity, observe that some element $j$ must be added into $S$, otherwise, the main loop will break immediately.
	Thus, Algorithm \ref{alg: packing} has at most $n$ rounds.
	In each round, the selection of $j$ needs $O(n)$ queries, and the \textsc{Delete} procedure needs $O(n)$ queries since $S$ contains at most $n$ elements.
	Therefore, Algorithm \ref{alg: cardinality} uses $O(n^2)$ queries in total.
	
	Next, we consider the approximation ratio of the algorithm.
	
	First consider the case where $\sum_{i=1}^{m}b_iw_{i(t+1)}<e^{\epsilon W}$.
	By the reasoning of Lemma \ref{lem: feasible}, we know that $S_{t+1}$ is returned as a feasible solution.
	Beside, this case happens because $f(j\mid S_{t+1})\leq 0$ for any $j\in N\setminus S_{t+1}$.
	By submodularity, $f(O\mid S_{t+1})\leq \sum_{j\in O} f(j\mid S_{t+1})\leq 0$.
	Since $S_{t+1}$ is the output of \textsc{Delete} in the last round, by Corollary \ref{cor: symmetric},
	\[ f(O)-2\cdot f(S_{t+1})\leq f(S_{t+1}\cup O)-f(S_{t+1})\leq 0. \]
	Thus, $f(S_{t+1})\geq f(O)/2$, which proves the lemma.
	
	Next, consider the case where $\sum_{i=1}^{m}b_iw_{i(t+1)}\geq e^{\epsilon W}$.
	For every $r\in [t]$, we have
	\begin{align*}
		\beta_r =\sum_{i=1}^{m} b_iw_{ir}&=\sum_{i=1}^{m} b_iw_{i(r-1)}\cdot (e^{\epsilon W})^{A_{ij_r}/b_i} \\
		&\leq \sum_{i=1}^{m} b_iw_{i(r-1)}\cdot\left(1+\frac{\epsilon W A_{ij_r}}{b_i}+\left(\frac{\epsilon W A_{ij_r}}{b_i}\right)^2\right) \\
		&\leq \sum_{i=1}^{m} b_iw_{i(r-1)}+(\epsilon W+\epsilon^2 W)\sum_{i=1}^{m}A_{ij_r}w_{i(r-1)} \\
		&\leq \beta_{r-1}\cdot\left(1+\frac{(\epsilon W+\epsilon^2 W)(f(S_r)-f(S_{r-1}))}{f(O)-2f(S_{r-1})}\right) \\
		&\leq \beta_{r-1}\cdot\exp \left(\frac{(\epsilon W+\epsilon^2 W)(f(S_r)-f(S_{r-1}))}{f(O)-2f(S_{r-1})}\right).
	\end{align*}
    The first inequality holds since $e^x\leq 1+x+x^2$ for $x\in[0,1]$ and $WA_{ij_r}/b_i\leq 1$ by the definition of $W$.
    The second holds again by $WA_{ij_r}/b_i\leq 1$.
    The third is due to Lemma \ref{lem: marginal gain-packing}.
    The last follows from the fact that $1+x\leq e^x$ for $x\in\mathbb{R}$.
    By expanding the recurrence, we get
    \begin{align*}
    	\beta_t &\leq \beta_0\prod_{r=1}^{t}\exp \left(\frac{(\epsilon W+\epsilon^2 W)(f(S_r)-f(S_{r-1}))}{f(O)-2f(S_{r-1})}\right) \\
    	&\leq\exp\left(\epsilon^2 W+(\epsilon W+\epsilon^2 W)\sum_{r=1}^{t}\frac{f(S_{r})-f(S_{r-1})}{f(O)-2f(S_{r-1})}\right).
    \end{align*}
    The last inequality holds since $\beta_0=m\leq \exp(\epsilon^2 W)$.
    
    Then, we give a lower bound for $\beta_t$.
    By the definition of $\beta_t$,
    \[ \beta_t e^{\epsilon}=\sum_{i=1}^{m} b_i w_{it}\cdot (e^{\epsilon W})^{1/W}\geq\sum_{i=1}^{m} b_i w_{it}\cdot (e^{\epsilon W})^{A_{ij_t}/b_i}=\sum_{i=1}^{m}b_iw_{i(t+1)}\geq e^{\epsilon W}. \]
    The first inequality is due to the definition of $W$.
    Thus, we have $\beta_t\geq e^{\epsilon(W-1)}$.
    
    Next, it is easy to see that $f(S_0)\leq f(S_1)\leq \ldots \leq f(S_t)$ since the algorithm never added an element with a negative marginal value.
    Besides, we can assume that $f(O)>2f(S_t)$ since otherwise the lemma already holds.
    Thus, we can apply Lemma \ref{lem: expand recursion} to get
    \[ \frac{\epsilon(W-1)-\epsilon^2 W}{\epsilon W+\epsilon^2 W}\leq \sum_{\ell=1}^{t}\frac{f(S_{r})-f(S_{r-1})}{f(O)-2f(S_{r-1})}\leq \frac{1}{2}\ln \frac{f(O)-2f(S_0)}{f(O)-2f(S_t)}. \]
    Finally, note that $(\epsilon(W-1)-\epsilon^2 W)/(\epsilon W+\epsilon^2 W)\geq (1-2\epsilon)(1+\epsilon)\geq 1-3\epsilon$.
    One can obtain that
    \[ f(S_t)\geq\frac{1}{2}\left(1-e^{-2(1-3\epsilon)}\right)\cdot f(O). \]
\end{proof}

When $m=1$, Algorithm \ref{alg: packing} reduces to a greedy algorithm for the knapsack constraint.
By the standard \emph{enumeration} technique \cite{KhullerMN99,Sviridenko04}, it is easy to remove the large-width assumption.
We can easily obtain the following result for the knapsack constraint.
\begin{theorem}
	\label{thm: knapsack}
	For symmetric submodular maximization under a knapsack constraint, there is an algorithm that has an approximation ratio of $\frac{1}{2}(1-e^{-2})$ and uses $O(n^4)$ queries.
\end{theorem}

\section{Conclusion}
\label{sec: conclusion}

In this paper, we present efficient deterministic algorithms for maximizing symmetric submodular functions under various constraints.
All of them require fewer queries and most of them achieve state-of-the-art approximation ratios.
However, our fast algorithm for the cardinality constraint is randomized and linear only when $k=O(\sqrt{n})$.
It is interesting to design a linear deterministic algorithm for the constraint.

%

\section*{Acknowledgement}
This work was supported by the National Natural Science Foundation of China Grants No. 62325210, 62272441.

\bibliographystyle{plain}
\bibliography{symmetric}

\end{document}